\documentclass[11pt]{article}

\usepackage{bussproofs}
\usepackage[numbers]{natbib}
\usepackage[colorlinks,citecolor=blue,urlcolor=blue]{hyperref}
\usepackage[a4paper,lmargin=3cm,tmargin=3cm,bmargin=2cm,rmargin=2cm]{geometry}
\usepackage{amsthm}
\usepackage{amssymb}

\def\a{\AxiomC}
\def\u{\UnaryInfC}
\def\b{\BinaryInfC}
\def\t{\TrinaryInfC}
\def\n{\noLine}

\title{Revisiting the proof theory of Classical \textbf{S4}}
\author{Bruno Lopes$^1$ \and Cec\'\i ilia Englander$^2$ \and Fernanda Lobo$^3$ \and Marcela Cruz$^2$ \and \vspace{1em}\\
$^1$Instituto de Computa\c c\~ao\\Universidade Federal Fluminense\\\href{mailto://bruno@ic.uff.br}{bruno@ic.uff.br} \vspace{1em}\\
$^2$Departamento de Inform\'atica\\Pontif\'\i cia Universidade Cat\'olica do Rio de Janeiro\\\href{mailto://cenglander@inf.puc-rio.br}{cenglander@inf.puc-rio.br}, \href{mailto://mcruz@inf.puc-rio.br}{mcruz@inf.puc-rio.br}\vspace{1em}\\
$^3$Departamento de Filosofia\\Pontif\'\i cia Universidade Cat\'olica do Rio de Janeiro\\\href{mailto://lobo.fernanda@gmail.com}{lobo.fernanda@gmail.com}}
\date{September, 2015}

\begin{document}

\maketitle

\begin{abstract}
In 1965 Dag Prawitz presented an extension of Gentzen-type systems of Natural Deduction to modal concepts of \textbf{S4}.
Maria da Paz Medeiros showed in 2006 that the proof of normalisation for classical \textbf{S4} does not hold and proposed a new proof of normalisation for a logically equivalent system, the system \textbf{NS4}.
However two problems in the proof of the critical lemma used by Medeiros in her proof were pointed out by Yuuki Andou in 2009.
This paper presents a proof of the critical lemma, resulting in a proof of normalisation for \textbf{NS4}.
\end{abstract}
 
\section{Introduction}
 
In his Ph.D.\ thesis, Dag Prawitz \cite{Prawitz2006} extended the Gentzen-type systems of Natural Deduction (ND) to modal concepts, obtaining Gentzen-type systems of ND for \textbf{S4} based on classical, intuitionistic and minimal predicate logic.
For this purpose, a modal operator of necessity (here represented by $\Box$) was added together with the rules for its introduction and elimination.
Prawitz then presented three formalizations of those modal systems, which differed in the restrictions on the introduction rule for the $\Box$, and stated that only the third one would accept the Normalisation Theorem.
 
About forty years later, Maria da Paz Medeiros \cite{Medeiros2006} argued that Prawitz's normalisation procedure does not work on the third version of the ND system for classical \textbf{S4}, and proposed a new system for \textbf{S4}, the \textbf{NS4}, for which the Normalisation Theorem would hold.
 
However, recently Yuuki Andou \cite{Yuuki2009} pointed out two problems in the proof of a lemma (the critical lemma) that plays a crucial role in Medeiros' s proof of the Normalisation Theorem.
Andou presented a Normal Form Theorem~\cite{Yuuki2009}, showing that for any proof $\Pi$ there is a proof $\Pi'$ in normal form by means of cut-elimination, but do not present a normalisation procedure.

In this paper we present a correction of Medeiros's proof of the aforementioned lemma and fulfil a normalisation procedure for \textbf{NS4}, which gives a computational interpretation of proofs by means of the Curry-Howard Correspondence.

After some definitions used in the present work (Section~\ref{sec:definitions}), we outline the original third version of Prawitz's system for classical \textbf{S4} and the counterexamples by Medeiros (Section~\ref{sec:medeiros}).
We then discuss the two cases in which the system may not produce valid derivations on \textbf{NS4} due to problems in the proof of the critical lemma (Section~\ref{sec:yuuki}).
In Section~\ref{sec:we} we present a proof of the critical lemma for \textbf{NS4}.
Our concern here is with Classical Propositional \textbf{S4}, but an extension of Classical First Order \textbf{S4} could be easily obtained by adding the corresponding rules for quantifiers.

\section{Definitions}\label{sec:definitions}

Based in the work of Maria da paz Medeiros \cite{Medeiros2006}, we present the definitions used in this work.

\newtheorem{dfn}{Definition}
\begin{dfn}
The premisses ($A \to B$) of the rule ($E_\to$), ($A \land B$) of $(E_\land)$, 
($A \lor B$) of $(E_\lor)$, ($\Box A$) of $(E_\Box)$, and the premisses $\Box B_1,...,\Box B_n$
 of ($I_\Box$) are called \emph{major premisses} and the others \emph{minor premisses}.
\end{dfn}

\begin{dfn}
A \emph{segment} in a derivation is a sequence $A_1, ..., A_n$ of occurrences of the same formula in a branch such that $A_1$ is not the conclusion of an application of ($E_\lor$) nor a discharged assumption through an application of ($I_\Box$), and $A_n$ is not a minor premiss of ($E_\lor$) nor a major premiss of ($I_\Box$).
\end{dfn}

\begin{dfn}
The \emph{length of a segment} is the number of formula occurrences in this segment.
\end{dfn}

\begin{dfn}
A \emph{maximal segment} in a derivation is a segment $A_1,...,A_n$ such that $A_1$ is the conclusion of an application of an introduction rule or ($\bot_c$), and $A_n$ is a major premiss of an application of an elimination rule.
\end{dfn}

\begin{dfn}
A \emph{maximal formula} is a maximal segment whose length is 1(one). A premiss is called \emph{maximal premiss} if it belongs to some maximal segment.
\end{dfn}

\begin{dfn}
A formula A is a \emph{trivial formula} if A is the conclusion of an application of ($\bot_c$) and the minor premiss of an application of ($E_\to$) whose major premiss is the assumption $\neg A$.
\end{dfn}

\begin{dfn}
	The \emph{degree of a formula} A, $G(A)$, is the number of occurrences in $A$ of logical symbols different from $\bot$. The degree of a segment is the degree of the formula that belongs to this segment.
\end{dfn}

\begin{dfn}
The \emph{degree of a derivation} $\Pi$, $G(\Pi)$ is the highest degree of a maximal segment of $\Pi$. If $\Pi$ does not have maximal segments, then $G(\Pi)=0$.
\end{dfn}

\begin{dfn}
	A \emph{critical derivation} is a derivation $\Pi$ such that, if $G(\Pi) = d$, then the last inference of $\Pi$ has a maximal premiss with degree $d$, and for every subderivation $\Sigma$ of $\Pi$, $G(\Sigma) < G(\Pi)$.
\end{dfn}

\begin{dfn}
The \emph{index} of a derivation $\Pi$ is $I(\Pi)=\langle d,s\rangle$, where $s$ is the sum of the lengths of the maximal segments of $\Pi$ whose degree is $d$.
If $\Pi$ does not have maximal segments, then $I(\Pi)=\langle 0,0\rangle$.
\end{dfn}

\begin{dfn}
A derivation $\Pi$ is a \emph{normal derivation} if $\Pi$ does not have maximal segments.
\end{dfn}
    
\begin{dfn} 
A formula $A$ is \emph{essentially modal} when each occurrence of a predicate parameter or predicate constant  in $A$ stands within the scope of an occurrence of $\Box$.
\end{dfn}

\section{A counterexample for Prawitz's classical \textbf{S4} system}\label{sec:medeiros}
 
According to the restriction on the $\Box$-I rule in Prawitz's third version of {\textbf{S4}}, if a formula $A$ depends on an assumption $B$ and there exists an essentially modal formula $F$ on the thread of $A$ from $B$ such that $A$ depends on every assumption which $F$ depends on, then $\Box$-I could be applied at $A$.

But Maria da Paz Medeiros \cite{Medeiros2006} argued that such restriction would not avoid maximal formulas by pointing out that although the following derivation
 \begin{prooftree}
  \a{$[\neg\Box A]^1$}\a{$\neg\Box A\to B$}\b{$B$}\a{$\neg B$}\b{$\bot$}\RightLabel{\scriptsize{1}}\u{$\Box A$}\u{$A$}\u{$C\to A$}\u{$\Box(C\to A)$}
 \end{prooftree}
is valid in {\textbf{S4}}, its reduction is not:
 \begin{prooftree}
   \a{$[\Box A]^1$}\u{$A$}
   \a{$[\neg A]^2$}\RightLabel{\scriptsize{1}}\b{$\bot$}\u{$\neg\Box A$}
 \a{$\neg\Box A\to B$}\b{$B$}\a{$\neg B$}\b{$\bot$}\RightLabel{\scriptsize{2}}\u{$A$}\u{$C\to A$}\u{$\Box(C\to A)$}
 \end{prooftree}
 
It was presented a new Natural Deduction for {\textbf{S4}} system, called system {\textbf{NS4}} and proposed a normalisation proof by critical lemma \cite{Medeiros2006}.
This new system is composed of the logical symbols $\land, \lor, \to, \bot\mbox{ and }\Box$ and the usual rules, except for $\Box$-I, which is as follows:
\begin{prooftree}
 \a{$\Box\vec{B}$}\a{$[\Box\vec{B}]^k$}\n\u{$\Lambda$}\n\u{$A$}\RightLabel{\scriptsize{$k$}}\b{$\Box A$}
\end{prooftree}
 
(We use \a{$\vec{F}$}\n\u{$\Box \vec{B}$}\DisplayProof for a sequence of deductions of the form \a{$\Pi_i$}\n\u{$\Box B_i$}\DisplayProof,
where no two $\Box B_i$'s are equal.
We also write $[\Box \vec{B}]^k$ to indicate that each $\Box B_i$ is discharged at $k$.
Labels may be dropped.)

This restriction on $\Box$-I rule states that \emph{all} the assumptions in $[\Box\vec{B}]^k$ must be discharged by the application of $\Box$-I and the premiss $A$ must not depend on any assumption other than the ones in $\Box\vec{B}$. 
The reason for this is explained in item 4 of the proof of the critical lemma and it does not affect the completeness of the system.
 
Together with this new $\Box$-I rule, we have the following reduction:
\begin{eqnarray}\label{reduction}
  \a{$\vec{\Sigma}$}\n\u{$\Box\vec{B}$}\a{$[\Box\vec{B}]^k$}\n\u{$\Lambda$}\n\u{$A$}\RightLabel{\scriptsize{$k$}}\b{$\Box A$}\u{$A$}\DisplayProof
&\rhd&
  \a{$\vec{\Sigma}$}\n\u{$(\Box\vec{B})$}\n\u{$\Lambda$}\n\u{$A$}\DisplayProof
\end{eqnarray}
 
\noindent and the permutative reduction bellow:
 \begin{eqnarray}\label{permutation}
 \a{$\vec{\Sigma}$}\n\u{$\Box\vec{A}$}
 \a{$[\Box\vec{A}]^i$}\n\u{$\Lambda_1$}\n\u{$B$}
 \RightLabel{\scriptsize{$i$}}\b{$\Box B$}
 \a{$\vec\Psi$}\n\u{$\Box\vec D$}
 \a{$[\Box B]^j[\Box\vec D]^{k}$}\n\u{$\Lambda_{2}$}\n\u{$C$}
 \RightLabel{\scriptsize{$k,j$}}
 \t{$\Box C$}
 \DisplayProof\rhd\nonumber\\\nonumber\\\\
 \a{$\vec\Sigma$}\n\u{$\Box\vec A$}
 \a{$\vec\Psi$}\n\u{$\Box\vec D$}
 \a{$[\Box\vec A]^{j}$}
 \a{$[\Box\vec A]^{k}$}\n\u{$\Lambda_1$}\n\u{B}
 \RightLabel{\scriptsize{$k$}}
 \b{$(\Box B)$}
 \a{$[\Box\vec D]^i$}\n
 \b{$\Lambda_2$}\n\u{C}
 \RightLabel{\scriptsize{$j,i$}}
 \t{$\Box C$}
 \DisplayProof\nonumber
 \end{eqnarray}

\section{A problem in the normalisation proof of \textbf{NS4}}\label{sec:yuuki}
 
Medeiros's normalisation proof \cite{Medeiros2006} begins with the assumption that a derivation $\Pi$ of $C$ from $\Gamma$ can be transformed in a derivation $\Pi_0$. The aim is to show that $I(\Pi_0)<I(\Pi)$. Next, it uses a critical lemma according to which, if $\Pi$ is a critical derivation of $C$ from $\Gamma$, then $\Pi$ can be transformed into a derivation $\Pi'$ such that $I(\Pi')<I(\Pi)$.
By the critical lemma, a subderivation $\Sigma$ of $\Pi$ can be transformed in a subderivation $\Sigma'$ such that $I(\Sigma')<I(\Sigma)$; but, then, $\Pi_1$ is the derivation resulted from the substitution of $\Sigma'$ for $\Sigma$ in $\Pi_0$, and $I(\Pi_1)<I(\Pi_0)$. 

However, recently Yuuki \cite{Andou2009} pointed out two flaws in the proof of Medeiros's critical lemma. 
The first one concerns critical derivations of the form
 \begin{eqnarray}\label{andou:1}
 \Pi\equiv
 \a{$\Sigma_{0,1}$}\n\u{$F$}
 \a{$[\neg F]^i$}
 	\b{$\bot$}\n\u{$\Sigma_{0,2}$}\n\u{$\bot$}\RightLabel{\scriptsize $i$}\u{$F$}
 \a{$\vec\Sigma$}\n\u{$\vec H$}\RightLabel{\scriptsize $r$}\b{$C$}\DisplayProof
 \end{eqnarray}
 where the major premiss $F$ is the conclusion of $\bot_C$ and $r$ is an elimination rule.
 According to Medeiros, this derivation can be transformed into
 \begin{eqnarray}
 \Pi'\equiv
 \a{$\Sigma_{0,1}$}\n\u{$F$}
 \a{$\vec\Sigma$}\n\u{$\vec H$}\RightLabel{\scriptsize $r$}
 	\b{$\bot$}\n\u{$\Sigma_{0,2}$}\n\u{$\bot$}\DisplayProof
 &\mbox{ or to }&
 \Pi''\equiv
\a{$\Sigma_{0,1}$}\n\u{$F$}
\a{$\vec\Sigma$}\n\u{$\vec H$}\RightLabel{\scriptsize $r$}
	\b{$C$}\a{$[\neg C]^i$}\b{$\bot$}\n\u{$\Sigma_{0,2}$}\n\u{$\bot$}\RightLabel{\scriptsize $i$}\u{$C$}\DisplayProof
 \end{eqnarray} depending on $C$ being $\bot$ or not.
 
Note that the assumptions of the form $\neg F$ discharged at the rule $i$ may occur more than once in $\Pi$,
and that the premiss $F$ which is conclusion of $\Sigma_{0,1}$ may be a maximal premiss in $\Pi'$ and in 
$\Pi''$.
In this case the index of either $\Pi'$ or $\Pi''$ may be even greater than that of $\Pi$.
Besides, one of the $H_i$'s in $\vec H$, say $H_l$, may be a maximal formula of degree $G(F)$ and, in this case, even if the $F$ side connected with $\neg F$ is not a maximal formula in $\Pi'$, this $H_l$ still is and the induction hypothesis cannot be used.
 
The second problem pointed out is when $\Pi$ has degree $G(\Box A)$ and is a critical derivation of the form
 
 \begin{eqnarray}\label{andou:2}
 \Pi\equiv &
 \a{$\vec\Sigma$}\n\u{$\Box\vec B$}
 \a{$[\Box\vec B]^k$}\n\u{$\Lambda_1$}\n\u{$A$}\RightLabel{\scriptsize $k$}
 \b{$\Box A$}
 \a{$[\neg\Box A]^i$}
 \b{$\bot$}\n\u{$\Sigma_0$}\n\u{$\bot$}\RightLabel{\scriptsize $i$}\u{$\Box A$}
 \a{$\vec\Psi$}\n\u{$\Box\vec C$}
 \a{$[\Box A]^j[\Box C]^\ell$}\n\u{$\Lambda_2$}\n\u{$B$}
 \RightLabel{\scriptsize $\ell,j$}\t{$\Box B$}
 \DisplayProof
 \end{eqnarray}
 
 If $\Box A$ occurs more than once as top-formula of $\Lambda_2$, by reducing $\Pi$ to
 $$
 \Pi'\equiv
 \a{$\vec\Sigma$}\n\u{$\Box\vec B$}
 \a{$\vec\Psi$}\n\u{$\Box\vec C$}
 \a{$[\Box\vec B]^j$}
 \a{$[\Box\vec B]^k$}\n\u{$\Lambda_1$}\n\u{$A$}
 \RightLabel{\scriptsize $k$}\b{$\Box A$}
 \a{$[\Box C]^\ell$}
 \n\b{$\Lambda_2$}\n\u{$B$}
 \RightLabel{\scriptsize $\ell,j$}\t{$\Box B$}
 \a{$[\neg\Box B]^i$}
 \b{$\bot$}\n\u{$\Sigma_0$}\n\u{$\bot$}\RightLabel{\scriptsize $i$}\u{$\Box B$}
 \DisplayProof
 $$
 the number of occurrences of $\Box A$ as maximal formula in $\Pi'$ will be greater than in $\Pi$.
 
Thus, it is possible that the reduction process generates copies of maximal formulas, so the index of $\Pi'$ may be greater than that of $\Pi$.

\section{Yet another proof of the critical lemma}\label{sec:we}

We present a proof of the critical lemma for the fragment $\{\land,\to,\bot,\Box\}$.
Extensions to First Order Logic is only handwork.

\newtheorem{reescrita}{Lemma}
\begin{reescrita}
 A critical derivation of the form $\Pi\equiv$
\a{$[\neg F]^1$}\n\u{$\Sigma_1$}\n\u{$\bot$}\RightLabel{\scriptsize 1}\u{$F$}
\a{$\vec{\Sigma}$}\n\u{$\vec{B}$}
\RightLabel{\scriptsize $r$}\b{$C$}\DisplayProof
where $F$ is the conclusion of $\bot_c$, can be transformed in a derivation\\
$\Pi_1\equiv$
\a{$\Sigma'_{1,1}$}\n\u{$F$}
\a{$[\neg F]^1$}
\b{$\bot$}\n\u{$\Sigma'_{1,2}$}\n\u{$\bot$}\RightLabel{\scriptsize 1}\u{$F$}
\a{$\vec{\Sigma}$}\n\u{$\vec{B}$}
\RightLabel{\scriptsize $r$}\b{$C$}\DisplayProof
where $\Pi_1$ is a derivation without trivial formulas.
Thus, the end-formula of $\Sigma'_{1,1}$ is not conclusion of $\bot_c$.
\end{reescrita}
\begin{proof}
See the work of Medeiros \cite{Medeiros2006}.
\end{proof}

Note that $\Pi_1$ has no more maximal formulas of degree equal to or higher than $G(F)$ than $\Pi$.
We will use the symbol $\propto$ to indicate the transformation of a derivation into a derivation without trivial formulas.

\newtheorem{normalizacao}{Theorem}[section]
\begin{normalizacao}
  If $\Pi$ is a critical derivation of $C$ from $\Gamma$, then $\Pi$ can be transformed into a derivation $\Pi'$ such that $G(\Pi') < G(\Pi)$.
\end{normalizacao}
\begin{proof}
Suppose $\Pi$ is a critical derivation with maximal premisses of degree $G(\Pi)$ which are premisses of the last inference of $\Pi$, $\# G(\Pi)$ is the number of maximal formulas of $\Pi$ with degree $G(\Pi)$ and $\ell(\Pi)$ is the lenght of $\Pi$.
 The proof is by induction on the pair $\langle\# G(\Pi), \ell(\Pi)\rangle$.
\begin{enumerate}
 \item 
$\Pi\equiv$ \a{$\Sigma_1$}\n\u{$A$}\a{$\Sigma_2$}\n\u{$B$}\b{$A\land B$}\u{$A$}\DisplayProof
    $\rhd$ 
\a{$\Sigma_1$}\n\u{$A$}\DisplayProof
    $\equiv \Pi'$

It is easy to see that $G(\Pi')<G(\Pi)$.

 \item 
    $\Pi\equiv$
\a{$\Sigma_1$}\n\u{$A$}\a{$[A]$}\n\u{$\Sigma_2$}\n\u{$B$}\u{$A\to B$}\b{$B$}\DisplayProof
    $\rhd$
\a{$\Sigma_1$}\n\u{$(A)$}\n\u{$\Sigma_2$}\n\u{$B$}\DisplayProof
    $\equiv \Pi'$ 

It is easy to see that $G(\Pi')<G(\Pi)$.

 \item 
    $\Pi\equiv$
\a{$\vec{\Sigma}$}\n\u{$\Box \vec{B}$}
\a{$[\Box \vec{B}]$}\n\u{$\Lambda_1$}\n\u{$C$}\b{$\Box C$}\u{$C$}\DisplayProof
    $\rhd$
\a{$\vec{\Sigma}$}\n\u{$(\Box \vec{B})$}\n\u{$\Lambda_1$}\n\u{$C$}\DisplayProof
$\equiv \Pi'$

If there exists a $\Box B_l$ which is a maximal premiss at $\Pi'$, then it would be a maximal formula at $\Pi$ and, as $\Pi$ is a critical derivation, $G(\Box B_l) < G(\Box C)$.
Thus, $G(\Pi')<G(\Pi)$.

 \item\label{caso4} 
    $\Pi\equiv$
    \a{$\vec{\Sigma}$}\n\u{$\Box\vec{B}$}
		\a{$[\Box\vec{B}]^k$}\n\u{$\Lambda_1$}\n\u{$A$}
		\b{$\Box A$ \RightLabel{\scriptsize $k$}}  
\a{$\vec{\Psi}$}\n\u{$\vec{H}$}
\a{$[\Box A]^j, [\vec{H}]^l$}\n\u{$\Lambda_2$}\n\u{$C$}\RightLabel{\scriptsize $j,l$}
\t{$\Box C$}\DisplayProof
$\rhd$\\
\a{$\vec{\Sigma}$}\n\u{$\Box\vec{B}$}
\a{$\vec{\Psi}$}\n\u{$\vec{H}$}
    \a{$[\Box\vec{B}]^l$}\n    
 \a{$[\Box\vec{B}]^k$}\n\u{$\Lambda_1$}\n\u{$A$}\RightLabel{\scriptsize $k$}
	\b{$(\Box A)$}
    \a{$[\vec{H}]^j$}\n
	\b{$\Lambda_2$}\n\u{$C$}\RightLabel{\scriptsize $j,l$}
\t{$\Box C$}
\DisplayProof
$\equiv\Pi_1$

We have two cases to consider:
\begin{enumerate}
 \item\label{subcaso4a} 
 \emph{there is an occurrence of $\Box A$ which is top-formula of $\Lambda_2$ and major premiss of an application of $\Box$-E}:
in this case, the number of maximal formulas of degree $G(\Box A)$ in $\Pi_1$ may be even greater than that of $\Pi$, as there may be more than one occurrence of $\Box A$ as top-formula of $\Lambda_2$.

There is a critical subderivation $\Xi_1$ of $\Pi_1$ of the form
\a{$\Box\vec{B}$}\n
 \a{$[\Box\vec{B}]^k$}\n\u{$\Lambda_3$}\n\u{$A$}\RightLabel{\scriptsize $k$}
       \b{$\Box A$}\u{$A$}\DisplayProof
which can be reduced to  $\Xi'_1\equiv$\a{$\Box\vec{B}$}\n\u{$\Lambda_3$}\n\u{$A$}\DisplayProof (case 3).

 \item
 \emph{there is an occurrence of $\Box A$ which is top-formula of $\Lambda_2$ and major premiss of $\Box$-I} :
 then, there is a critical subderivation $\Xi_2$ of the form

\a{$\Box\vec{B}$}
\a{$[\Box\vec{B}]^k$}\n\u{$\Lambda_3$}\n\u{$A$}\RightLabel{\scriptsize $k$}
		  \b{$\Box A$}
\a{$\vec{\Upsilon}$}\n\u{$\Box\vec{C}$}
\a{$[\Box A]^p, [\Box\vec{C}]^q$}
\n\u{$\Lambda_{4}$}\n\u{$D$}\RightLabel{\scriptsize $p,q$}\t{$\Box D$}\DisplayProof

The lenght of $\Xi_2$ is smaller than the lenght of $\Pi$.
Hence, by induction hypothesis, we can reduce $\Xi_2$ to a $\Xi'_2$ such that $G(\Xi'_2)<G(\Pi)$.

Note that we cannot guarantee that the lenght of $\Xi_2$ is smaller than the lenght of $\Pi$ if there were more than one occurrence of $\Box A$ as top-formula of $\Lambda_2$ in $\Pi_1$, and if there were many occurrences of $\Box A$ as major premiss in $\Xi_2$.
That is the reason of the restriction on the beginning of the section.
\end{enumerate}

Let $\Pi_2$ be the result of replacing each occurrence of critical subderivations of the form of $\Xi_1$ and the form of $\Xi_2$ in $\Pi_1$ by $\Xi'_1$ and $\Xi'_2$ respectively.

If $\Box A$ is the only major premiss that is maximal formula in $\Pi$, i.e., there is no member of $\vec{H}$ which is a maximal premiss of the same degree of $\Pi$, then $G(\Pi_2)<G(\Pi)$ and $\Pi_2=\Pi$.
Otherwise, i.e., if there exists a $m$ such that $H_m$ is a maximal formula in $\Pi$, then $\#G (\Pi_2)< \# G(\Pi)$ and, as $\Pi_2$ is a critical derivation, by induction hypothesis $\Pi_2$ can be transformed into a derivation $\Pi'$ such that $G(\Pi')<G(\Pi_2)$.
Hence, as $\Pi$ was transformed into $\Pi_2$ and $G(\Pi_2)$ is not higher than $G(\Pi)$, $G(\Pi')<G(\Pi)$.

\item 
  $\Pi\equiv$
\a{$[\neg (A\land B)]^1$}\n\u{$\Sigma_1$}\n\u{$\bot$}\RightLabel{\scriptsize 1}\u{$A\land B$}\u{$A$}\DisplayProof
$\propto$
\a{$\Sigma'_{1,1}$}\n\u{$A\land B$}
\a{$[\neg (A\land B)]^1$}
\b{$\bot$}\n\u{$\Sigma'_{1,2}$}\n\u{$\bot$}\RightLabel{\scriptsize 1}\u{$A\land B$}\u{$A$}\DisplayProof
$\rhd$
\a{$\Sigma'_{1,1}$}\n\u{$A\land B$}\u{$A$}
\a{$[\neg A]^1$}
\b{$\bot$}\n\u{$\Sigma'_{1,2}$}\n\u{$\bot$}\RightLabel{\scriptsize 1}\u{$A$}\DisplayProof
    $\equiv\Pi_1$

If the end formula of $\Sigma'_{1,1}$ is not the conclusion of an introduction rule, then the end-formula of $\Sigma'_{1,1}$ is not a maximal formula and $G(\Pi_1) < G(\Pi)$ and $\Pi_1 \equiv \Pi'$.
If the end formula of $\Sigma'_{1,1}$ is the conclusion of an introduction rule, then $\Pi_1$ is of the form
\a{$\Sigma_3$}\n\u{$A$}
\a{$\Sigma_4$}\n\u{$B$}
  \b{$A\land B$}\u{$A$}
\a{$[\neg A]^1$}
\b{$\bot$}\n\u{$\Sigma'_{1,2}$}\n\u{$\bot$}\RightLabel{\scriptsize 1}\u{$A$}\DisplayProof
    which can be reduced to
\a{$\Sigma_3$}\n\u{$A$}
\a{$[\neg A]^1$}
\b{$\bot$}\n\u{$\Sigma'_{1,2}$}\n\u{$\bot$}\RightLabel{\scriptsize 1}\u{$A$}\DisplayProof
  $\equiv\Pi'$ and $G(\Pi')<G(\Pi)$.\\\\\\

\item 
  $\Pi\equiv$
\a{$[\neg (A\to B)]$}\n\u{$\Sigma_1$}\n\u{$\bot$}\u{$A\to B$}
\a{$\Sigma_2$}\n\u{$A$}
\b{$B$}\DisplayProof
$\propto$
\a{$\Sigma'_{1,1}$}\n\u{$A\to B$}
\a{$[\neg (A\to B)]^1$}
  \b{$\bot$}\n\u{$\Sigma'_{1,2}$}\n\u{$\bot$}\RightLabel{\scriptsize 1}\u{$A\to B$}
\a{$\Sigma_2$}\n\u{$A$}
  \b{$B$}\DisplayProof
$\rhd$\\\\\\

\a{$\Sigma'_{1,1}$}\n\u{$A\to B$}
\a{$\Sigma_{2}$}\n\u{$A$}
\b{$B$}
\a{$[\neg B]^1$}
\b{$\bot$}\n\u{$\Sigma'_{1,2}$}\n\u{$\bot$}\RightLabel{\scriptsize 1}\u{$B$}\DisplayProof
    $\equiv\Pi_1$

If the end formula of $\Sigma'_{1,1}$ is not the conclusion of an introduction rule, then $ G(\Pi_1) < G(\Pi)$ and $\Pi_1\equiv \Pi'$.
If the end formula of $\Sigma'_{1,1}$ is the conclusion of an introduction rule, then $\Pi_1$ is of the form

\a{$[A]$}\n\u{$\Sigma_3$}\n\u{$B$}\u{$A\to B$}
\a{$\Sigma_4$}\n\u{$A$}
  \b{$B$}
\a{$[\neg B]^1$}
\b{$\bot$}\n\u{$\Sigma'_{1,2}$}\n\u{$\bot$}\RightLabel{\scriptsize 1}\u{$B$}\DisplayProof
    which can be reduced to
\a{$\Sigma_4$}\n\u{$(A)$}\n\u{$\Sigma_3$}\n\u{$B$}
\a{$[\neg B]^1$}
\b{$\bot$}\n\u{$\Sigma'_{1,2}$}\n\u{$\bot$}\RightLabel{\scriptsize 1}\u{$B$}\DisplayProof
  $\equiv\Pi'$ and $G(\Pi')<G(\Pi)$.\\\\\\

 \item 
    $\Pi\equiv$
\a{$[\neg \Box A]$}\n\u{$\Sigma_1$}\n\u{$\bot$}\u{$\Box A$}\u{$A$}\DisplayProof
    $\propto$
\a{$\Sigma'_{1,1}$}\n\u{$\Box A$}
\a{$[\neg \Box A]^1$}
\b{$\bot$}\n\u{$\Sigma'_{1,2}$}\n\u{$\bot$}\RightLabel{\scriptsize 1}\u{$\Box A$}\u{$A$}\DisplayProof
    $\rhd$
\a{$\Sigma'_{1,1}$}\n\u{$\Box A$}\u{$A$}
\a{$[\neg A]^1$}
\b{$\bot$}\n\u{$\Sigma'_{1,2}$}\n\u{$\bot$}\RightLabel{\scriptsize 1}\u{$A$}\DisplayProof
$\equiv\Pi_1$

If the end formula of $\Sigma'_{1,1}$ is not the conclusion of an introduction rule, then $ G(\Pi_1) <  G(\Pi)$ and $\Pi_1\equiv \Pi'$.
If the end formula of $\Sigma'_{1,1}$ is the conclusion of an introduction rule, then $\Pi_1$ is of the form

\a{$\vec{\Psi}$}\n\u{$\Box\vec{B}$}
\a{$[\Box\vec{B}]^k$}\n\u{$\Lambda_1$}\n\u{$A$}
  \RightLabel{\scriptsize $k$}\b{$\Box A$}\u{$A$}
\a{$[\neg A]^l$}
\b{$\bot$}\n\u{$\Sigma'_{1,2}$}\n\u{$\bot$}\RightLabel{\scriptsize $k,l$}\u{$A$}\DisplayProof
    which can be reduced to
\a{$\vec{\Psi}$}\n\u{$(\Box\vec{B})$}\n\u{$\Lambda_1$}\n\u{$A$}
  \a{$[\neg A]^l$}
\b{$\bot$}\n\u{$\Sigma'_{1,2}$}\n\u{$\bot$}\RightLabel{\scriptsize $l$}\u{$A$}\DisplayProof
  $\equiv\Pi_2$

If one of the $\Box B_i$'s, say $\Box B_m$, were a maximal formula in $\Pi_2$, it would be a maximal formula in $\Pi$ and, as $\Pi$ is a critical derivation, $G(\Box B_m) < G(\Box A)$. Thus, $G(\Pi_2)<G(\Pi)$ and $\Pi_2\equiv\Pi'$.

\item 
    $\Pi\equiv$
\a{$[\neg\Box A]^k$}\n\u{$\Sigma_1$}\n\u{$\bot$}\RightLabel{\scriptsize $k$}\u{$\Box A$}
\a{$\vec{\Psi}$}\n\u{$\vec{H}$}
\a{$[\Box A]^l,[\vec{H}]^j$}\n\u{$\Sigma_2$}\n\u{$C$}\RightLabel{\scriptsize $l,j$}
  \t{$\Box C$}\DisplayProof
    $\propto$\\\\

\a{$\Sigma'_{1,1}$}\n\u{$\Box A$}
\a{$[\neg\Box A]^k$}\b{$(\bot)$}
\n\u{$\Sigma'_{1,2}$}\n\u{$\bot$}\RightLabel{\scriptsize $k$}\u{$\Box A$}
\a{$\vec{\Psi}$}\n\u{$\vec{H}$}
\a{$[\Box A]^l, [\vec{H}]^j$}\n\u{$\Sigma_2$}\n\u{$C$}\RightLabel{\scriptsize $l,j$}
\t{$\Box C$}\DisplayProof
$\rhd$\\\\\\

\a{$\Sigma'_{1,1}$}\n\u{$\Box A$}
\a{$\vec{\Psi}$}\n\u{$\vec{H}$}
\a{$[\Box A]^l,[\vec{H}]^j$}\n\u{$\Sigma_2$}\n\u{$C$}\RightLabel{\scriptsize $j,l$}
\t{$\Box C$}
\a{$[\neg\Box C]^k$}
\b{$(\bot)$}\n\u{$\Sigma'_{1,2}$}\n\u{$\bot$}\RightLabel{\scriptsize $k$}\u{$\Box C$}\DisplayProof
$\equiv\Pi_1$

Note that $\Sigma'_{1,1}$ is a subderivation of $\Sigma_1$.
Hence, if the subderivation
$\Lambda\equiv$
\a{$\Sigma'_{1,1}$}\n\u{$\Box A$}
\a{$\vec{\Psi}$}\n\u{$\vec{H}$}
\a{$[\Box A]^l,[\vec{H}]^j$}\n\u{$\Sigma_2$}\n\u{$C$}\RightLabel{\scriptsize $l,j$}
\t{$\Box C$}\DisplayProof
of $\Pi_1$
is a critical derivation, its lenght is smaller than the lenght of $\Pi$.
Thus, by the induction hypothesis, $\Lambda$ can be reduced to a derivation $\Lambda'$ such that
$G(\Lambda')<G(\Pi)$.
The result of replacing each occurrence of $\Lambda$ in $\Pi_1$ by $\Lambda'$ is a derivation $\Pi'$ such that $G(\Pi')<G(\Pi)$.

\item 
$\Pi\equiv$
\a{$[\neg F]^1$}\n\u{$\Sigma_1$}\n\u{$\bot$}\RightLabel{\scriptsize 1}\u{$F$}
\a{$\vec{\Psi}$}\n\u{$\vec{H}$}
\RightLabel{\scriptsize $r$}\b{$\bot$}\DisplayProof
$\propto$
\a{$\Sigma'_{1,1}$}\n\u{$F$}
\a{$[\neg F]^1$}
\b{$\bot$}\n\u{$\Sigma'_{1,2}$}\n\u{$\bot$}\RightLabel{\scriptsize 1}\u{$F$}
\a{$\vec{\Psi}$}\n\u{$\vec{H}$}
\RightLabel{\scriptsize $r$}\b{$\bot$}\DisplayProof$\rhd$
$\Pi_1\equiv$
\a{$\Sigma'_{1,1}$}\n\u{$F$}
\a{$\vec{\Psi}$}\n\u{$\vec{H}$}
\b{$\bot$}\n\u{$\Sigma'_{1,2}$}\n\u{$\bot$}\DisplayProof

The critical subderivation $\Lambda\equiv$ 
\a{$\Sigma'_{1,1}$}\n\u{$F$}
\a{$\vec{\Psi}$}\n\u{$\vec{H}$}
\b{$\bot$}\DisplayProof
 of $\Pi_1$ is smaller than $\Pi$.
Thus, by the induction hypothesis, $\Lambda$ can be reduced to a derivation $\Lambda'$ such that
$G(\Lambda')<G(\Pi)$.
By replacing each occurrence of $\Lambda$ in $\Pi_1$ by $\Lambda'$ we achieve the desired derivation.
This case deals with classical $\bot$ with the elimination of implication, conjunction and box.

\item 
$\Pi\equiv$
\a{$[\neg F]^1$}
\a{$\Sigma'$}\n\u{$\neg F\to C$}
\b{$C$}\n\u{$\Sigma$}\n\u{$\bot$}\RightLabel{\scriptsize 1}\u{$F$}\RightLabel{\scriptsize  el}\u{$B$}\DisplayProof
$\rhd\Pi'\equiv$
\a{$[F]^1$}\u{$B$}
\a{$[\neg B]^2$}\b{$\bot$}\RightLabel{\scriptsize 1}\u{$\neg F$}
\a{$\Sigma'$}\n\u{$\neg F\to C$}
\b{$C$}\n\u{$\Sigma$}\n\u{$\bot$}\RightLabel{\scriptsize 2}\u{$B$}\DisplayProof

If $\neg F\to C$ were a maximal formula in $\Pi'$, it would be a maximal formula in $\Pi$, which is not the case as, by hypothesis, $G(\Pi) = G(F)$ and $G(F)<G(\neg F\to C)$.
Hence, $G(\Pi)<G(\Pi')$.
\end{enumerate}
\end{proof}

\section{Conclusions}

This work presented the problem pointed out by Maria da Paz Medeiros \cite{Medeiros2006} on the normalisation procedure proposed by Dag Prawitz \cite{Prawitz2006}, followed by the system that she proposed \cite{Medeiros2006}, the \textbf{NS4}.
She presented a normalisation proof for this system for which we presented the problems pointed out by Yuuki Andou \cite{Andou2009} and finally we presented a proof of the Normalisation Theorem for \textbf{NS4}.

Among other deductive systems for $\textbf{S4}$, there are some where the Normalisation Theorem holds, like Sequent Calculus.
There is also a Natural Deduction with Labels system by Alex Simpson \cite{Simpson1994} for which the Normalisation Theorem holds.
But the system proposed by Dag Prawitz \cite{Prawitz2006} and Maria da Paz Medeiros \cite{Medeiros2006} are pure Natural Deduction systems, without semantic interferences (as the labels from the system of Alex Simpson) for which there are no previous proof of the Normalisation Theorem known into the available literature.
Yuuki Andou showed that for any proof of {\bf S4} there is a normalised proof via cut-elimination~\cite{Yuuki2009} but did not present a normalisation procedure.
We fulfilled it by presenting a correction in Medeiros' proof that lead to a normalised Natural Deduction system for {\bf S4}, the {\bf NS4} system.

\section*{Acknowledgements}

The authors thank CNPq, CAPES and FAPERJ for partially supporting this work. The authors would like also to thank Prof.\ Luiz Carlos Pereira (Departamento de Filosofia, Pontif\'\i cia Universidade Cat\'olica do Rio de Janeiro) for his advice and insights.

\bibliographystyle{agsm}
\bibliography{ref}
 
\end{document}